\documentclass[11pt,a4paper]{article}

\DeclareMathAlphabet{\mathbbmsl}{U}{bbm}{m}{sl}
\newcommand{\HH}{\mathbbmsl H}
\newcommand{\FF}{\mathbbmsl F}
\usepackage{fullpage}
\usepackage{amsthm}
\usepackage{amssymb}
\usepackage{enumerate}
\usepackage{palatino}
\usepackage{graphicx}
\usepackage{color}

\newtheorem{theorem}{Theorem}[section]
\newtheorem{lemma}[theorem]{Lemma}

\title{A Faster Algorithm to Recognize Even-Hole-Free Graphs\thanks{To
    appear in {\em Journal of Combinatorial Theory, Series B}. The
    current version slightly improves upon the preliminary
    version~\cite{ChangL12} appeared in SODA~2012:~(a) The time
    complexity for recognizing even-hole-free $n$-node $m$-edge graphs
    $G$ is reduced from $O(m^2 n^7)$ to $O(m^3 n^5)$, which is an
    improvement if $m=o(n^2)$; and (b) if $G$ contains even holes,
    then the current version shows how to output an even hole of $G$
    also in $O(m^3 n^5)$ time.}}
    
\author{%
Hsien-Chih Chang\thanks{Department of Computer Science, University of
  Illinois at Urbana-Champaign, USA. This research was performed while
  this author was affiliated with Department of Computer Science and
  Information Engineering, National Taiwan University. Email:
  {\tt{hchang17@illinois.edu}}.}
\and
Hsueh-I Lu\thanks{Corresponding author. Department of Computer
  Science and Information Engineering, National Taiwan University.
  This author also holds joint appointments in the Graduate Institute
  of Networking and Multimedia and the Graduate Institute of
  Biomedical Electronics and Bioinformatics, National Taiwan
  University. Address: 1 Roosevelt Road, Section 4, Taipei 106,
  Taiwan, ROC. Research of this author is supported in part by NSC
  grant~101--2221--E--002--062--MY3.  Email:
  {\tt{hil@csie.ntu.edu.tw}}. Web: {\url{www.csie.ntu.edu.tw/\~hil}}.}
}

\begin{document}

\maketitle

\begin{abstract}
We study the problem of determining whether an $n$-node graph $G$
contains an {\em even hole}, i.e., an induced simple cycle consisting
of an even number of nodes.  Conforti, Cornu{\'e}jols, Kapoor, and
Vu\v{s}kovi\'{c} gave the first polynomial-time algorithm for the
problem, which runs in $O(n^{40})$ time.  Later, Chudnovsky,
Kawarabayashi, and Seymour reduced the running time to $O(n^{31})$.
The best previously known algorithm for the problem, due to da~Silva
and Vu\v{s}kovi\'{c}, runs in $O(n^{19})$ time.  In this paper, we
solve the problem in $O(n^{11})$ time via a decomposition-based
algorithm that relies on the decomposition theorem of da~Silva and
Vu\v{s}kovi\'{c}.  Moreover, if $G$ contains even holes, then our
algorithm also outputs an even hole of $G$ in $O(n^{11})$ time.
\end{abstract}

\paragraph{\em Keywords} 
even hole, decomposition-based detection algorithm, extended clique
tree, $2$-join, star-cutset, diamond, beetle, tracker,


\section{Introduction}
\label{section:intro}
For any graphs $G$ and $F$, we say that $G$ {\em contains} $F$ if $F$
is isomorphic to an induced subgraph of~$G$.  If $G$ does not contain
$F$, then $G$ is {\em $F$-free}.  For any family $\FF$ of graphs, $G$
is {\em $\FF$-free} if $G$ is $F$-free for each graph $F$ in $\FF$.  A
{\em hole} is an induced simple cycle consisting of at least four
nodes.  A hole is {\em even} (respectively, {\em odd}) if it consists
of an even (respectively, odd) number of nodes.  Even-hole-free graphs
have been extensively studied in the literature~(see,
e.g.,~\cite{ConfortiCKV99, ConfortiCKV00, ConfortiCKV02a, daSilvaV07,
  Addario-BerryCHRS08, SilvaSS10, daSilvaV13, KloksMV09}).  See
Vu\v{s}kovi\'{c}~\cite{Vuskovic10} for a recent survey.  This paper
studies the problem of determining whether a graph contains even
holes.  Let $n$ be the number of nodes of the input graph.  Conforti,
Cornu{\'e}jols, Kapoor, and
Vu\v{s}kovi\'{c}~\cite{ConfortiCKV97con,ConfortiCKV02b} gave the first
polynomial-time algorithm for the problem, which runs in $O(n^{40})$
time~\cite{ChudnovskyKS05}.  Later, Chudnovsky, Kawarabayashi, and
Seymour~\cite{ChudnovskyKS05} reduced the running time to $O(n^{31})$.
Chudnovsky et al.~\cite{ChudnovskyKS05} also observed that the running
time can be further reduced to $O(n^{15})$ as long as prisms can be
detected efficiently, but Maffray and Trotignon~\cite{MaffrayT05}
showed that detecting prisms is NP-hard.  The best previously known
algorithm for the problem, due to da~Silva and
Vu\v{s}kovi\'{c}~\cite{daSilvaV13}, runs in $O(n^{19})$ time.  We
solve the problem in $O(n^{11})$ time, as stated in the following
theorem.

\begin{theorem}
\label{theorem:theorem1}
It takes $O(m^3 n^5)$ time to determine whether an $n$-node $m$-edge
connected graph contains even holes.
\end{theorem}

\paragraph{Technical overview}

The $O(n^{40})$-time algorithm of Conforti et
al.~\cite{ConfortiCKV02b} is based on their decomposition
theorem~\cite{ConfortiCKV02a} stating that a connected even-hole-free
graph either (i) is an extended clique tree or (ii) contains non-path
$2$-joins or $k$-star-cutsets with $k\in \{1,2,3\}$. The main body of
their algorithm recursively decomposes the input graph $G$ into a list
$\mathbbmsl{L}$ of a polynomial number of smaller or simpler graphs
using non-path $2$-joins or $k$-star-cutsets with $k\in \{1,2,3\}$
until each graph in $\mathbbmsl{L}$ does not contain any of the
mentioned cutsets.  Since even holes in extended clique trees can be
detected in polynomial time, it suffices for their algorithm to ensure
the even-hole-preserving condition of $\mathbbmsl{L}$: $G$ is
even-hole-free if and only if all graphs in $\mathbbmsl{L}$ are
even-hole-free.  To ensure the condition of $\mathbbmsl{L}$, their
algorithm requires a cleaning process to either detect an even hole in
$G$ or remove bad structures from $G$ before obtaining $\mathbbmsl{L}$
from $G$.  The $O(n^{31})$-time algorithm Chudnovsky et
al.~\cite{ChudnovskyKS05}, which is not based upon any decomposition
theorem but still requires the cleaning process, looks for even holes
directly. (The algorithms of Chudnovsky et al.~\cite{ChudnovskyCLSV05}
for recognizing perfect graphs are also of this type of
non-decomposition-based algorithms.)
The $O(n^{19})$-time algorithm of da~Silva and
Vu\v{s}kovi\'{c}~\cite{daSilvaV13}, adopting the decomposition-based
approach, relies on a stronger decomposition theorem stating that if a
connected even-hole-free graph has no star-cutsets and non-path
$2$-joins, then it is an extended clique tree.  Since $k$-star-cutsets
with $k\in\{2,3\}$ need not be taken into account, the decomposition
process is significantly simplified, leading to a much lower time
complexity.
Our $O(n^{11})$-time algorithm is also based on the brilliant
decomposition theorem of da~Silva and
Vu\v{s}kovi\'{c}~\cite{daSilvaV13}.  Our improvement is obtained via
the new idea of trackers, which allow for fewer graphs to be generated
in the process of decomposition using star-cutsets.  The cleaning
process is also sped up by an algorithm for recognizing beetle-free
graphs, based upon the three-in-a-tree algorithm of Chudnovsky and
Seymour~\cite{ChudnovskyS10}.

Specifically, our recognition algorithm for even-hole-free graphs
consists of two phases.
Throughout the paper, a {\em $k$-hole} (respectively, {\em $k$-cycle}
and {\em $k$-path}) is a $k$-node hole (respectively, cycle and path).
The first phase (see~Lemma~\ref{lemma:cleaning}) either (1) ensures
that the input graph $G$ contains even holes via the existence of a
``beetle'' (see~\S\ref{section:prelim} and
Figure~\ref{figure:fig2}(a)) or a $4$-hole in $G$ or (2) produces a
set~$\mathbbmsl{T}$ of ``trackers'' $(H,u_1u_2u_3)$ of $G$, where $H$
is a beetle-free and $4$-hole-free induced subgraph of $G$ and
$u_1u_2u_3$ is a $3$-path of $H$. $\mathbbmsl{T}$ satisfies the
following even-hole-preserving condition (see Condition~L1): If~$G$
contains even holes, then at least one element $(H, u_1u_2u_3)$ of
$\mathbbmsl{T}$ is ``lucky'' such that a shortest even hole $C$ of $G$
is a subgraph of $H$ and the following holds: (a) $u_1u_2u_3$ is a
path of $C$ and (b) the neighborhood of $C$ in $H$ is ``super clean''
(i.e., $M_H(C)=N^{2,2}_H(C)=N^{1,2}_H(C)=N^4_H(C)=\varnothing$ using
notation defined in~\S\ref{section:prelim}).  The second phase applies
an algorithm (see Lemma~\ref{lemma:decomposition}) on each tracker
$(H,u_1u_2u_3)\in\mathbbmsl{T}$ to either ensure that $H$ contains
even holes or ensure that $(H,u_1u_2u_3)$ is not lucky.  If all
trackers in $\mathbbmsl{T}$ are not lucky, then the
even-hole-preserving condition of $\mathbbmsl{T}$ implies that $G$ is
even-hole-free.  Otherwise, an induced subgraph~$H$ of $G$ contains an
even hole, implying that $G$ contains even holes.

The recognition algorithm for beetle-free graphs (see the proof of
Lemma~\ref{lemma:cleaning}) in the first phase is based on Chudnovsky
and Seymour's three-in-a-tree algorithm~\cite{ChudnovskyS10} (see
Theorem~\ref{lemma:three_tree}). If $G$ contains beetles or $4$-holes,
then $G$ contains even holes. Otherwise, if $G$ contains even holes,
then the neighborhood of each shortest even hole $C$ of $G$ is
``clean'' (i.e., $N^{1,2}_G(C)=N^4_G(C)=\varnothing$, see the proof of
Lemma~\ref{lemma:hole-node}). To further ensure that the neighborhood
of $C$ is super clean, we generate a set $\mathbbmsl{S}$ of ``super
cleaners'' $(S,u_1u_2u_3)$, where $S$ is a node subset of $G$ and
$u_1u_2u_3$ is a path of $G$, such that at least one super cleaner
$(S,u_1u_2u_3)\in\mathbbmsl{S}$ satisfies $u_1u_2u_3\subseteq
C\subseteq H=G\setminus S$ and $M_H(C)=N^{2,2}_H(C)=\varnothing$ for
some shortest even hole $C$ of $G$ (see the proof of
Lemma~\ref{lemma:cleaning}). The set $\mathbbmsl{T}$ consisting of the
trackers $(G\setminus S, u_1u_2u_3)$ with $(S,
u_1u_2u_2)\in\mathbbmsl{S}$ satisfies the required
even-hole-preserving condition. The underlying technique of guessing
``bad nodes'' of $G$ (using
Lemmas~\ref{lemma:short-prism},~\ref{lemma:major},
and~\ref{lemma:gate}) to be removed by super cleaners is called
``cleaning'' in the literature (see, e.g.,
Vu\v{s}kovi\'{c}~\cite[\S4]{Vuskovic10}).

The algorithm applied on each tracker $T=(H,
u_1u_2u_3)\in\mathbbmsl{T}$ in the second phase relies on the
decomposition theorem of da~Silva
and~Vu\v{s}kovi\'{c}~\cite{daSilvaV13} (see
Theorem~\ref{lemma:structure-new}).  Since even holes can be
efficiently detected in an extended clique tree (see
Lemma~\ref{lemma:extended-clique-tree-even-hole}, which is a slightly
faster implementation of the algorithm of da~Silva
and~Vu\v{s}kovi\'{c}~\cite{daSilvaV13}), our algorithm performs two
stages of even-hole-preserving decompositions on $H$, first via
star-cutsets and then via non-path $2$-joins, until each of the
resulting graphs either is an extended clique tree or has $O(1)$
nodes.  If all of the resulting graphs are even-hole-free, then $T$ is
not lucky; otherwise, $H$ contains even holes.  As noted by
Chv\'{a}tal~\cite{Chvatal85} (see Lemma~\ref{lemma:fullstar}), if $H$
has no dominated nodes, then a star-cutset of $H$ has to be a full
star-cutset of $H$, which can be efficiently detected.  Thus, at the
beginning of each decomposition in the first stage, we preprocess
$(H,u_1u_2u_3)$ by deleting all dominated nodes of $H$ and carefully
updating nodes $u_1$, $u_2$, and $u_3$ such that the luckiness of
$(H,u_1u_2u_3)$ is preserved (see Lemma~\ref{lemma:dominate-find}).
The correctness of this preprocessing step relies on the fact that $H$
is beetle-free and the requirement for $(H,u_1u_2u_3)$ to be lucky
that the neighborhood of some shortest even hole $C$ in $H$ with
$u_1u_2u_3\subseteq C$ is super clean.  Path $u_1u_2u_3$ is crucial in
the stage of decompositions via star-cutsets for the graph $H$ having
no dominated nodes.  Specifically, if $S$ is a star-cutset of $H$,
then by merely examining the neighborhood of path $u_1u_2u_3$ in $H$,
we can efficiently identify a connected component $B$ of $H\setminus
S$ such that $(H[S\cup B], u_1u_2u_3)$ preserves the luckiness of
$(H,u_1u_2u_3)$ (see Step~\ref{dominated3} in the proof of
Lemma~\ref{lemma:star-cutset}).  We then let $H=H[C\cup B]$ and repeat
the above procedure for $O(n)$ iterations until $H$ has no
star-cutsets.  The second stage, i.e., decompositions via non-path
$2$-joins for graphs having no star-cutsets, is based upon the
detection algorithm for non-path $2$-joins of Charbit et
al.~\cite{CharbitHTV12} (see Lemma~\ref{lemma:non-path-2-join}).  This
stage decomposes an $m$-edge graph having no star-cutsets into a set
of $O(m)$ smaller graphs, each of which either consists of~$O(1)$
nodes or is an extended clique tree (see the proof of
Lemma~\ref{lemma:no-star-cutsets}).

\paragraph{Related work}
Even-hole-free planar graphs~\cite{Porto92con} can be recognized in
$O(n^3)$ time.  It is NP-complete to determine whether a graph
contains an even (respectively, odd) hole that passes a given
node~\cite{Bienstock91,Bienstock92}.  The Strong Perfect Graph Theorem
of Chudnovsky, Robertson, Seymour, and Thomas~\cite{ChudnovskyRST06}
states that a graph $G$ is perfect if and only if both $G$ and the
complement of $G$ are odd-hole-free.  Although perfect graphs can be
recognized in $O(n^9)$ time~\cite{ChudnovskyCLSV05}, the tractability
of recognizing odd-hole-free graphs remains open (see,
e.g.,~\cite{Johnson05}).  Polynomial-time algorithms for detecting odd
holes are known for planar graphs~\cite{Hsu87}, claw-free
graphs~\cite{ShremSG10, KennedyK11un}, and graphs with bounded clique
numbers~\cite{ConfortiCLVZ06}.  Graphs containing no holes (i.e.,
chordal graphs) can be recognized in $O(m+n)$
time~\cite{TarjanY84,TarjanY85, RoseT75con, RoseTL76}.  Graphs
containing no holes consisting of five or more nodes (i.e., weakly
chordal graphs) can be recognized in $O(m^2+n)$
time~\cite{NikolopoulosP04con, NikolopoulosP07}.  It takes $O(n^2)$
time to detect a hole that passes any $o((\log n / \log \log
n)^{2/3})$ given nodes in an $O(1)$-genus
graph~\cite{KawarabayashiK09con, KawarabayashiK10}.
See~\cite{ChudnovskyRST10, Zwols10, Defossez09, ConfortiCV04} for more
results on odd-hole-free graphs.

\paragraph{Road map}
The rest of the paper is organized as follows.
Section~\ref{section:prelim} gives the preliminaries and proves
Theorem~\ref{theorem:theorem1} by Lemmas~\ref{lemma:cleaning}
and~\ref{lemma:decomposition}.  Section~\ref{section:cleaning} proves
Lemma~\ref{lemma:cleaning}.  Section~\ref{section:decomposition}
proves Lemma~\ref{lemma:decomposition}. Section~\ref{section:conclude}
concludes the paper by explaining how to augment our proof of
Theorem~\ref{theorem:theorem1} into an~$O(m^3 n^5)$-time algorithm
that outputs an even hole of an $n$-node $m$-edge graph containing
even holes.


\section{Preliminaries and the topmost structure of our proof}
\label{section:prelim}

Unless clearly specified otherwise, all graphs throughout the paper
are simple and undirected.  Let $|S|$ denote the cardinality of set
$S$.  Let $G$ be a graph.  Let $V(G)$ consist of the nodes in $G$.
For any subgraph $H$ of $G$, let $G[H]$ denote the subgraph of $G$
induced by $V(H)$.  Subgraphs $H$ and $H'$ of graph $G$ are {\em
  adjacent} in $G$ if some node of $H$ and some node of $H'$ are
adjacent in $G$.  For any subset~$U$ of $V(G)$, let $G\setminus
U=G[V(G)\setminus U]$.  For any subgraph $H$ of $G$, let $N_G(H)$
consist of the nodes of $V(G)\setminus V(H)$ that are adjacent to $H$
in~$G$ and let $N_G[H]=N_G(H)\cup V(H)$.

\begin{figure}[t]
\centerline{\scalebox{0.9}{\input{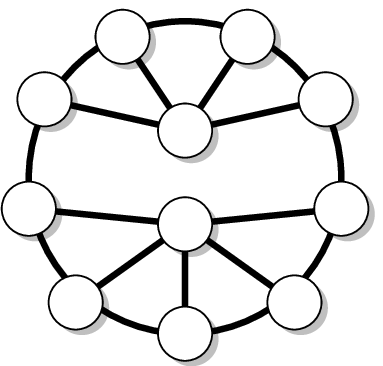}}}
\caption{$C=v_1c_2c_3v_2c_7c_8v_1$ is a clean even hole of the
  $11$-node graph $G$, since $M_G(C)=N_G^{2,2}(C)=\varnothing$.
  $C'=c_1c_2\cdots c_9c_1$ is an odd hole with $M_G(C')=\{v_2\}$.}
\label{figure:fig1}
\end{figure}

Let $C$ be a hole of $G$.  Let $x$ be a node in $V(G)\setminus V(C)$.
Let $N_C(x)=N_G(x)\cap V(C)$. We say that $x$ is a {\em major
  node}~\cite{ChudnovskyKS05} of $C$ in $G$ if at least three distinct
nodes of $N_C(x)$ are pairwise non-adjacent in $G$.  Let $M_G(C)$
consist of the major nodes of~$C$ in $G$.  For instance, in
Figure~\ref{figure:fig1}, $M_G(C)=\varnothing$ and $M_G(C')=\{v_2\}$.

\begin{lemma}[{Chudnovsky et al.~{\cite[Lemma 2.2]{ChudnovskyKS05}}}]
\label{lemma:hole-parity}
If $C$ is a shortest even hole of graph $G$ and $x\in M_G(C)$, then 
$|N_C(x)|$ is even.
\end{lemma}

\noindent
If $x\in N_G(C)\setminus M_G(C)$, then $1\leq |N_C(x)|\leq 4$ and
$C[N_C(x)]$ has at most two connected components.  Moreover, if
$C[N_C(x)]$ is not connected, then each connected component of
$C[N_C(x)]$ has at most two nodes.  Let $N_G^{i}(C)$ with $1\leq i\leq
4$ consist of the nodes $x\in N_G(C)\setminus M_G(C)$ such that
$|N_C(x)|=i$ and $C[N_C(x)]$ is connected.  Let $N_G^{i,j}(C)$ with
$1\leq i\leq j\leq 2$ consist of the nodes $x\in N_G(C)\setminus
M_G(C)$ such that $C[N_C(x)]$ is not connected and the two connected
components of $C[N_C(x)]$ has $i$ and $j$ nodes, respectively.  We
have
\begin{equation}
\label{equation:partition}
N_G(C)=N_G^{1}(C)\cup N_G^{2}(C)\cup N_G^{3}(C)\cup N_G^{4}(C)\cup
N_G^{1,1}(C)\cup N_G^{1,2}(C)\cup N_G^{2,2}(C)\cup M_G(C).
\end{equation}
We say that $C$ is a {\em clean} hole of $G$ if
$M_G(C)=N_G^{2,2}(C)=\varnothing$.  We say that $C$ is a {\em
  $u_1u_2u_3$-hole} of $G$ if $u_1u_2u_3$ is a $3$-path of $C$ and $C$
is a clean shortest even hole of $G$.  For instance, if $G$ is as
shown in Figure~\ref{figure:fig1}, then $C=v_1c_2c_3v_2c_7c_8v_1$ is a
$v_1c_2c_3$-hole of $G$.  If $H$ is an induced subgraph of $G$ and
$u_1u_2u_3$ is a $3$-path of $H$, then we call $(H,u_1u_2u_3)$ a {\em
  tracker} of $G$.  A tracker~$(H,u_1u_2u_3)$ of $G$ is {\em lucky} if
there is a $u_1u_2u_3$-hole of $H$.  If there are lucky trackers of
$G$, then $G$ contains even holes. Therefore, a set $\mathbbmsl{T}$ of
trackers of $G$ satisfying the following even-hole-preserving
condition reduces the problem of determining whether $G$ is
even-hole-free to the problem of determining whether all trackers in
$\mathbbmsl{T}$ are not lucky:
\begin{enumerate}[\em {Condition~L}1:]
\item 
If $G$ contains even holes, then at least one element of
$\mathbbmsl{T}$ is a lucky tracker of $G$.
\end{enumerate}

\begin{figure}[t]
\centerline{\scalebox{0.9}{\input{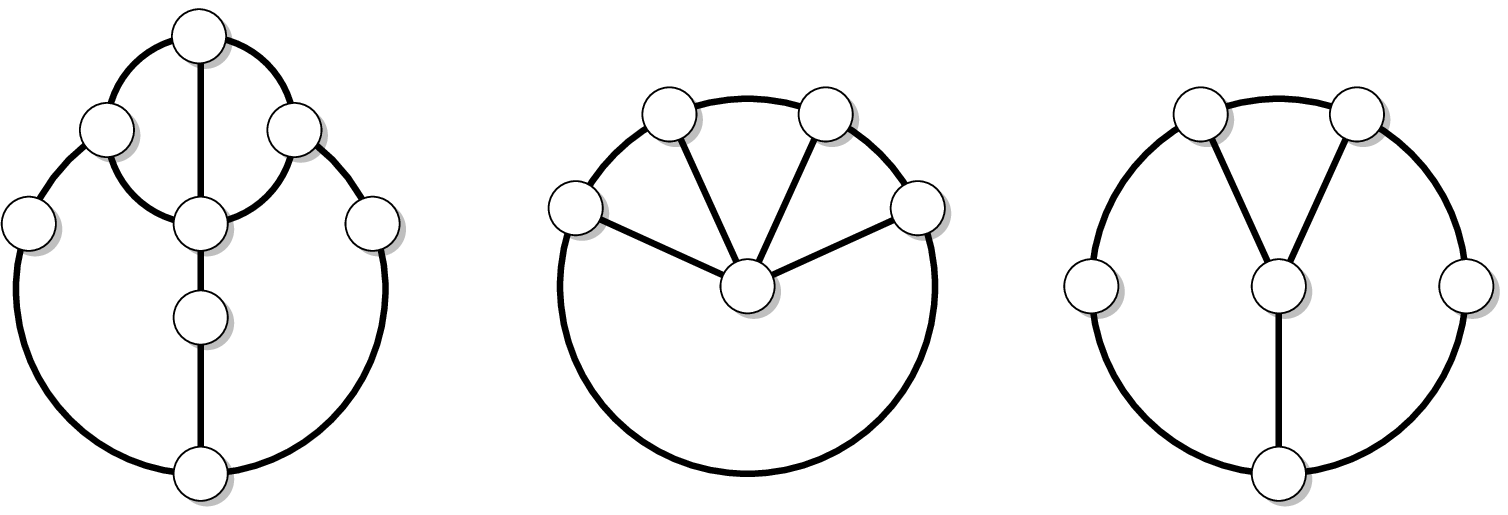}}}
\caption{(a) A beetle $B$, where $B[\{b_1,b_2,b_3,b_4\}]$ is a
  diamond. (b) If $x\in N_G^4(C)$, then $G[C\cup\{x\}]$ is a beetle
  $B$, where $B[\{u_1,u_2,u_3,x\}]$ is a diamond. (c) A node $x\in
  N_G^{1,2}(C)$.}
\label{figure:fig2}
\end{figure}

An induced subgraph $B$ of $G$ is a {\em beetle} of $G$ if $B$
consists of
(1) a $4$-cycle $b_1b_2b_3b_4b_1$ with exactly one chord $b_2b_4$
(i.e., a {\em diamond}~\cite{KloksMV09,ConfortiCKV02b} of $G$) and
(2) a tree $I$ of $G\setminus\{b_4\}$ having exactly three
leaves~$b_1$, $b_2$, and $b_3$ with the property that
$I\setminus\{b_1,b_2,b_3\}$ is an induced tree of $G$ not adjacent to
$b_4$.  See Figure~\ref{figure:fig2}(a) for an illustration.  Node
$b_5$ (respectively, $b_6$ and $b_7$) is the neighbor of $b_1$
(respectively, $b_2$ and $b_3$) in $I$.  Node $b_8$ is the only
degree-$3$ node of $I$.
Note that at least one of the three cycles in $B\setminus\{b_2\}$,
$B\setminus\{b_1,b_4\}$, and $B\setminus \{b_3,b_4\}$ is an even hole
of $G$.
Nodes $b_5$, $b_6$, $b_7$, and $b_8$ need not be distinct.  For
instance, as illustrated by Figure~\ref{figure:fig2}(b), if $C$ is a
hole of $G$ and $x$ is a node of $N_G^4(C)$, then $G[C\cup \{x\}]$ is
a beetle of~$G$.

\begin{lemma}
\label{lemma:hole-node}
If $G$ is a beetle-free graph, then $N_G(C) \subseteq N_G^{1,1}(C)\cup
N_G^1(C)\cup N_G^2(C)\cup N_G^3(C)$ holds for any clean shortest even
hole $C$ of $G$.
\end{lemma}
\begin{proof}
By $M_G(C)=N_G^{2,2}(C)=\varnothing$ and
Equation~(\ref{equation:partition}), it suffices to show
$N_G^{1,2}(C)=N_G^{4}(C)=\varnothing$.  If~$x\in N_G^4(C)$ as
illustrated by Figure~\ref{figure:fig2}(b), then $G[C \cup \{x\}]$ is
a beetle of $G$, a contradiction.  If~$x\in N_G^{1,2}(C)$, then let
$u$, $v_1$, and $v_2$ be the nodes of $N_C(x)$ such that $v_1$ and
$v_2$ are adjacent in $C$, as illustrated by
Figure~\ref{figure:fig2}(c).  Let $P_1$ be the path of~$C\setminus
\{v_2\}$ between $u$ and $v_1$.  Let $P_2$ be the path of~$C\setminus
\{v_1\}$ between $u$ and $v_2$.  Either $G[\{x\}\cup P_1]$ or
$G[\{x\}\cup P_2]$ is an even hole of $G$ shorter than~$C$, a
contradiction.
\end{proof}

\subsection{Proving Theorem~\ref{theorem:theorem1}}

\begin{lemma}
\label{lemma:cleaning}
It takes $O(m^3n^5)$ time to complete either one of the following
tasks for any $n$-node $m$-edge graph $G$. Task~1: Ensuring that $G$
contains even holes.  Task~2: (a) Ensuring that $G$ is beetle-free and
(b) obtaining a set $\mathbbmsl{T}$ of $O(m^2 n)$ trackers of $G$ that
satisfies Condition~L1.
\end{lemma}

\begin{lemma}
\label{lemma:decomposition}
Given a tracker $T=(H,u_1u_2u_3)$ of an $n$-node $m$-edge beetle-free
graph $G$, it takes $O(mn^4)$ time to either ensure that $H$ contains
even holes or ensure that $T$ is not lucky.
\end{lemma}

\begin{proof}[Proof of Theorem~\ref{theorem:theorem1}]
We apply Lemma~\ref{lemma:cleaning} on $G$ in $O(m^3n^5)$ time.  If
Task~1 is completed, then the theorem is proved.  If Task~2 is
completed, then $G$ is beetle-free and we have a set $\mathbbmsl{T}$
of $O(m^2 n)$ trackers of~$G$ that satisfies Condition~L1.  By
Condition~L1 of~$\mathbbmsl{T}$ and Lemma~\ref{lemma:decomposition},
one can determine whether $G$ contains even holes in time
$|\mathbbmsl{T}|\cdot O(mn^4)=O(m^3 n^5)$.
\end{proof}


\section{Proving Lemma~\ref{lemma:cleaning}}
\label{section:cleaning}

A {\em clique} of $G$ is a complete subgraph of $G$.  A clique of $G$
is {\em maximal} if it is not contained by other cliques of $G$.  We
need the following theorem and three lemmas to prove
Lemma~\ref{lemma:cleaning}.

\begin{theorem}[Chudnovsky and Seymour~\cite{ChudnovskyS10}]
\label{lemma:three_tree}
Let $z_1$, $z_2$, and $z_3$ be three nodes of an $n$-node graph. It
takes $O(n^4)$ time to determine whether the graph contains an induced
tree $I$ with $\{z_1,z_2,z_3\}\subseteq V(I)$.
\end{theorem}

\begin{lemma}[Farber~{\cite[Proposition~2]{Farber89}} and da~Silva and Vu\v{s}kovi\'{c}~{\cite{daSilvaV07}}]
\label{lemma:clique1}
Let $G$ be an $n$-node $m$-edge $4$-hole-free graph.  It takes
$O(mn^2)$ time to either ensure that $G$ contains even holes or obtain
all $O(n^2)$ maximal cliques of $G$.
\end{lemma}

\begin{lemma}[da~Silva and Vu\v{s}kovi\'{c}~{\cite{daSilvaV07}}]
\label{lemma:clique2}
The number of maximal cliques in an $n$-node $m$-edge even-hole-free graph
is at most $n+2m$.
\end{lemma}

\begin{lemma}[Chudnovsky, Kawarabayashi, and Seymour~{\cite[Lemma~4.2]{ChudnovskyKS05}}]
\label{lemma:short-prism}
For any shortest even hole $C$ of a $4$-hole-free graph $G$, there is
an edge $v_1v_2$ of $C$ with $N_G^{2,2}(C) \subseteq N_G(v_1) \cap
N_G(v_2)$.
\end{lemma}

\begin{lemma}
\label{lemma:major}
For any shortest even hole $C$ of a $4$-hole-free graph $G$, if
$G[M_G(C)]$ is not a clique of $G$, then there is a node $u$ of $C$
with $M_G(C)\subseteq N_G(u)$.
\end{lemma}

Before working on the proof of Lemma~\ref{lemma:major}, we first prove
Lemma~\ref{lemma:cleaning} using Theorem~\ref{lemma:three_tree} and
Lemmas~\ref{lemma:clique1},~\ref{lemma:clique2},~\ref{lemma:short-prism},
and~\ref{lemma:major}.

\begin{proof}[Proof of Lemma~\ref{lemma:cleaning}]
We claim that $G$ contains beetles if and only if at least one of the
$O(m^3n)$ choices of node $b_4$ and three distinct edges $b_1b_5$,
$b_2b_6$, and $b_3b_7$ of $G$ satisfies all of the following four
conditions:
\begin{itemize}
\addtolength{\itemsep}{-0.5\baselineskip}
\item $G[\{b_1,b_2,b_3,b_4\}]$ is the $4$-cycle $b_1b_2b_3b_4b_1$ with
  exactly one chord $b_2b_4$.
\item The edges between $\{b_1,b_2,b_3\}$ and $\{b_5,b_6,b_7\}$ are
  exactly $b_1b_5$, $b_2b_6$, and $b_3b_7$.
\item $\{b_5,b_6,b_7\}\cap \{b_1,b_2,b_3,b_4\}=\varnothing$, but nodes
  $b_5$, $b_6$, and $b_7$ need not be distinct.
\item There is an induced tree $I'$ of $G\setminus
  \left((N_G[b_1]\cup\cdots\cup N_G[b_4])\setminus
  \{b_5,b_6,b_7\}\right)$ with $\{b_5,b_6,b_7\}\subseteq V(I')$.
\end{itemize}
The claim can be verified by seeing that if $I''$ is the minimal
subtree of $I'$ satisfying $\{b_5,b_6,b_7\}\subseteq V(I'')$, then
$I=I''\cup\{b_1b_5,b_2b_6,b_3b_7\}$ is a tree of $G\setminus\{b_4\}$
with leaf set $\{b_1,b_2,b_3\}$ having the property that $I\setminus
\{b_1,b_2,b_3\}$ is an induced tree of $G$ not adjacent to $b_4$.  By
the claim above and Theorem~\ref{lemma:three_tree}, it takes~$O(m^3
n^5)$ time to determine whether $G$ contains beetles.  It takes
$O(n^4)$ time to determine whether $G$ contains $4$-holes.  If $G$
contains $4$-holes or beetles, then $G$ contains even holes.  The
lemma is proved by completing Task~1 in $O(m^3 n^5)$ time.  The rest
of the proof assumes that $G$ is $4$-hole-free and beetle-free.

By Lemma~\ref{lemma:clique1}, it takes $O(m n^2)$ time to either
ensure that $G$ contains even holes or obtain the $O(n^2)$ maximal
cliques of $G$.  If $G$ contains even holes, then the lemma is proved
by completing Task~1 in $O(m n^2)$ time.  Otherwise, we have all the
$O(n^2)$ maximal cliques of $G$. If the number of maximal cliques in
$G$ is larger than $n+2m$, then Lemma~\ref{lemma:clique2} implies that
$G$ contains even holes, also proving the lemma by completing
Task~1. If the number of maximal cliques in $G$ is $n+2m$ or fewer,
then let $\mathbbmsl{T}$ consist of the trackers of $G$ that are in
the form of~$(G\setminus S_1, u_1u_2u_3)$ or~$(G\setminus
S_2,u_1u_2u_3)$ with
\begin{displaymath}
\begin{array}{rllll}
S_1 &=& S_1(u_1,u_2,u_3,v_1,v_2)&=&
(N_G(v_1) \cap N_G(v_2)) \cup (N_G(u_2) \setminus \{u_1,u_3\});\\
S_2 &=& S_2(u_1,u_2,K)&=&
(N_G(u_1) \cap N_G(u_2)) \cup V(K),
\end{array}
\end{displaymath}
where $u_1u_2$ and $v_1v_2$ are edges of $G$ and $K$ is a maximal
clique of $G$.  We have $|\mathbbmsl{T}| = O(m^2 n)$.  Since all
$O(n+m)$ maximal cliques of $G$ are available, $\mathbbmsl{T}$ can be
computed in time $O(m^2 n) \cdot O(n+m)=O(m^3 n)$ time.  To ensure the
completion of Task~2, it remains to prove that $\mathbbmsl{T}$
satisfies Condition~L1. Suppose that $G$ contains even holes.  Let $C$
be an arbitrary shortest even hole of $G$.  The following case
analysis shows that there are lucky trackers of $G$ in
$\mathbbmsl{T}$.

{\em Case~1: $M_G(C) \subseteq N_G(u_2)$ holds for a node $u_2$ of
  $C$}.  Let $u_1$ and $u_3$ be the neighbors of $u_2$ in~$C$. By
$M_G(C) \subseteq N_G(u_2) \setminus \{u_1,u_3\}$ and
Lemma~\ref{lemma:short-prism}, there is an edge $v_1v_2$ of $C$ with
$M_G(C)\cup N_G^{2,2}(C)\subseteq S_1$.  By the choices of $u_1$ and
$u_3$, we have $(N_{G}(u_2)\setminus \{u_1,u_3\})\cap C=\varnothing$.
Since $v_1v_2$ is an edge of hole $C$, we have $N_G(v_1) \cap N_G(v_2)
\cap C=\varnothing$.  Thus, $S_1\cap C=\varnothing$, implying that $C$
is a clean hole of $G\setminus S_1$ and $u_1u_2u_3$ is a path of $C$.
Since $C$ is a shortest even hole of $G$, $C$ is also a shortest even
hole of $G\setminus S_1$.  Therefore, $C$ is a $u_1u_2u_3$-hole of
$G\setminus S_1$.

{\em Case~2: $M_G(C)\not\subseteq N_G(u)$ holds for all nodes $u$ of
  $C$}.  By Lemma~\ref{lemma:major}, $G[M_G(C)]$ is a clique of~$G$.
Let $K$ be a maximal clique of $G$ with $M_G(C)\subseteq V(K)$.
Combining with Lemma~\ref{lemma:short-prism}, there is an edge
$u_1u_2$ of $C$ with $M_G(C)\cup N_G^{2,2}(C)\subseteq S_2$.  We have
$V(K)\cap C= \varnothing$ or else $M_G(C)\cap C=\varnothing$ implies
$M_G(C) \subseteq V(K) \setminus \{u\} \subseteq N_G(u)$ for any node
$u\in V(K) \cap C$, a contradiction.  Since $u_1u_2$ is an edge of
$C$, we have $N_G(u_1) \cap N_G(u_2)\cap C=\varnothing$. Thus,
$S_2\cap C=\varnothing$, implying that $C$ is a clean hole of
$G\setminus S_2$. Letting $u_3$ be the neighbor of $u_2$ in $C$ other
than $u_1$, $u_1u_2u_3$ is a path of $C$.  Since $C$ is a shortest
even hole of $G$, $C$ is also a shortest even hole of $G\setminus
S_2$.  Therefore, $C$ is a $u_1u_2u_3$-hole of $G\setminus S_2$.
\end{proof}

\begin{figure}[t]
\centerline{\scalebox{0.9}{\input{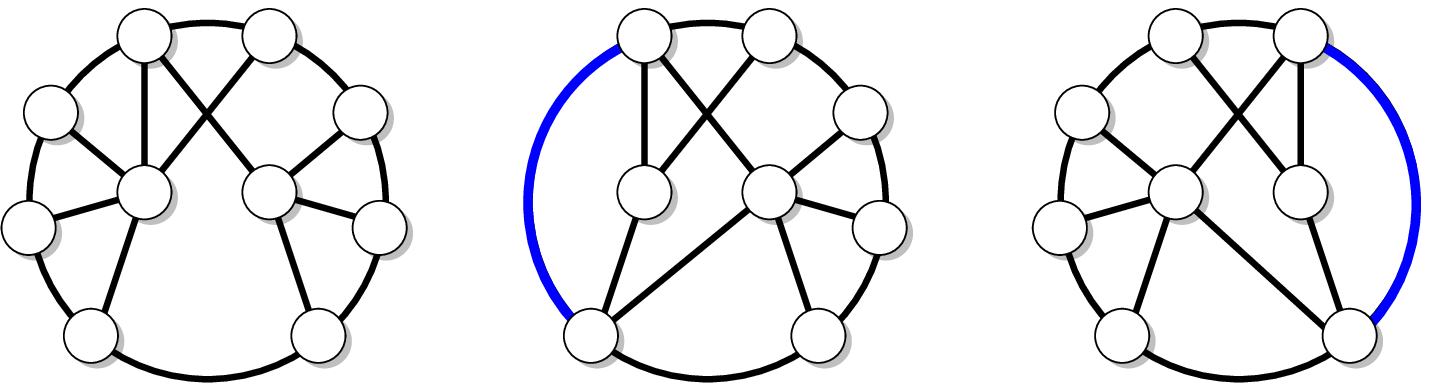}}}
\caption{(a) Edge $u_1u_2$ is a gate of the $8$-hole $C$ induced by
  nodes other than $x_1$ and $x_2$, which are the major nodes of
  $C$. (b) and (c) Illustrations for the proof of
  Lemma~\ref{lemma:major}.}
\label{figure:fig3}
\end{figure}

The rest of the section proves Lemma~\ref{lemma:major}.  An edge
$u_1u_2$ of hole $C$ is a {\em gate}~\cite{ChudnovskyKS05} of $C$ with
respect to major nodes $x_1$ and $x_2$ of $C$ if both of the following
conditions hold:
\begin{enumerate}[\em {Condition~G}1:]
\addtolength{\itemsep}{-0.5\baselineskip}

\item 
\label{item:G1}
There are two edges $u_1x_2$ and $u_2x_1$ and at least one of edges
$u_1x_1$ and $u_2x_2$.

\item 
\label{item:G2}
There is a node $u_0$ of $C\setminus\{u_1,u_2\}$ such that $x_1$
(respectively, $x_2$) is not adjacent to $C\setminus V(P_1)$
(respectively, $C\setminus V(P_2)$), where $P_1$ (respectively, $P_2$)
is the path of $C$ between $u_2$ (respectively, $u_1$) and $u_0$ that
passes $u_1$ (respectively, $u_2$).
\end{enumerate}
See Figure~\ref{figure:fig3}(a) for an illustration.

\begin{lemma}[Chudnovsky et al.~{\cite[Lemmas~2.3 and~2.4]{ChudnovskyKS05}}]
\label{lemma:gate}
The following statements hold for any shortest even hole $C$ of a
$4$-hole-free graph $G$.
\begin{enumerate}
\addtolength{\itemsep}{-0.5\baselineskip}
\item
\label{gate:item1}
If $x_1$ and $x_2$ are non-adjacent nodes of $M_G(C)$, then there is a
gate of $C$ with respect to $x_1$ and $x_2$ in~$G$.

\item
\label{gate:item2}
If $X$ is a subset of $M_G(C)$ with $|X|=3$ such that $G[X]$ has
at most one edge, then $X\subseteq N_G(u)$ holds for some node~$u$ of
$C$.
\end{enumerate}
\end{lemma}

\begin{proof}[Proof of Lemma~\ref{lemma:major}]
Let $x_1$ and $x_2$ be two non-adjacent nodes of $M_G(C)$.  Let $U$
consist of the nodes~$u$ of $C$ that are adjacent to both of $x_1$ and
$x_2$.  By Lemma~\ref{lemma:gate}(\ref{gate:item1}), there is a gate
$u_1u_2$ of~$C$ with respect to $x_1$ and $x_2$.  We have
$\varnothing\ne U \subseteq \{u_1,u_2,u_0\}$, where $u_0$ is a node of
$C$ ensured by Condition~G\ref{item:G2}.  Assume $u_0\in U$.  By
Condition~G\ref{item:G1}, $u_0$ is adjacent to $u_1$ or $u_2$ in $G$
or else one of $u_1x_1u_0x_2u_1$ and $u_2x_1u_0x_2u_2$ would be a
$4$-hole of $G$.
If $u_0$ is adjacent to $u_1$ as illustrated by
Figure~\ref{figure:fig3}(b), then Condition~G\ref{item:G2} implies
$N_C(x_1)=\{u_0,u_1,u_2\}$, which contradicts with $x_1\in M_G(C)$.
If~$u_0$ is adjacent to $u_2$ as illustrated by
Figure~\ref{figure:fig3}(c), then Condition~G\ref{item:G2} implies
$N_G(x_2)=\{u_0,u_1,u_2\}$, which contradicts with $x_2\in M_G(C)$.
Therefore, $u_0\not\in U$, and thus $U\subseteq \{u_1,u_2\}$.  The
lemma holds trivially if $|M_G(C)|=2$.  To prove the lemma for
$|M_G(C)|\geq 3$, we first show the claim: ``{\em Each node~$x\in
  M_G(C)\setminus\{x_1,x_2\}$ is adjacent to $U$.}''  If one of $x_1$
and $x_2$ is not adjacent to $x$, then the claim follows from
Lemma~\ref{lemma:gate}(\ref{gate:item2}).  If both of $x_1$ and $x_2$
are adjacent to $x$, then each node $u\in U$ is adjacent to~$x$ in $G$
or else $ux_1xx_2u$ is a $4$-hole, a contradiction. The claim is
proved.

\begin{figure}[t]
\centerline{\scalebox{0.9}{\input{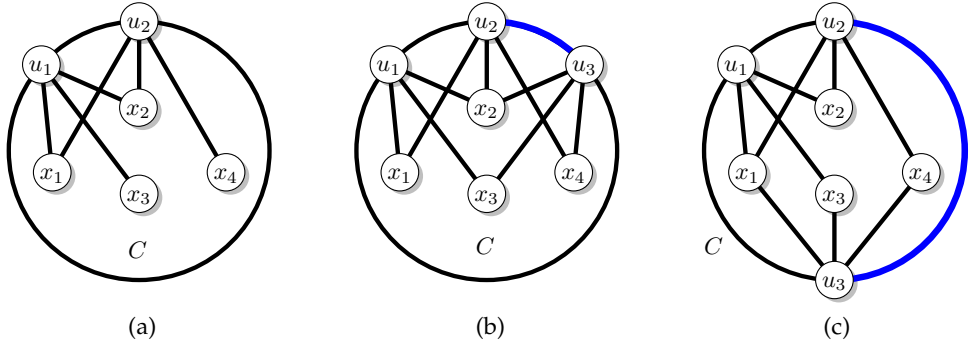}}}
\caption{Illustrations for the proof of Lemma~\ref{lemma:major}.}
\label{figure:fig4}
\end{figure}

By the claim above, the lemma holds if $|M_G(C)|=3$ or $|U|=1$.  It
remains to consider the cases with $|M_G(C)|\geq 4$ and
$U=\{u_1,u_2\}$ (thus, there are edges $u_1x_1$ and $u_2x_2$) by
showing that either $u_1$ or $u_2$ is adjacent to each node $x\in
M_G(C)$.  Assume $x_3\in M_G(C)\setminus N_G(u_2)$ and $x_4\in
M_G(C)\setminus N_G(u_1)$ for contradiction.  By the claim above, 
$u_1x_3$ and $u_2x_4$ are edges of $G$. We know $x_3\notin N_G(x_4)$ or
else $u_1u_2x_4x_3u_1$ is a $4$-hole.  See
Figure~\ref{figure:fig4}(a).  Observe that $x_4$ cannot be adjacent to
both of $x_1$ and $x_2$ or else $u_1x_1x_4x_2u_1$ is a $4$-hole.
{\em Case~1: $x_4$ is not adjacent to $x_2$}.  
By Lemma~\ref{lemma:gate}(\ref{gate:item2}), a node $u_3$ of $C$ is
adjacent to all of $x_2$, $x_3$, and $x_4$.  Since $u_3$ is adjacent
to both of $x_3$ and $x_4$, we have $u_3\notin\{u_1,u_2\}$.  See
Figure~\ref{figure:fig4}(b).  If $u_2u_3$ is an edge of $C$, then
$u_1x_3u_3u_2u_1$ is a $4$-hole; otherwise, $u_2x_2u_3x_4u_2$ is a
$4$-hole, a contradiction.
{\em Case~2: $x_4$ is not adjacent to $x_1$}.  
By Lemma~\ref{lemma:gate}(\ref{gate:item2}), a node $u_3$ of $C$ is
adjacent to all of~$x_1$, $x_3$, and $x_4$.  Since $u_3$ is adjacent
to both of $x_3$ and $x_4$, we have $u_3\notin\{u_1,u_2\}$.  See
Figure~\ref{figure:fig4}(c).  If $u_2u_3$ is an edge of $C$, then
$u_1x_3u_3u_2u_1$ is a $4$-hole; otherwise, $u_2x_1u_3x_4u_2$ is a
$4$-hole, a contradiction.
\end{proof}


\section{Proving Lemma~\ref{lemma:decomposition}}
\label{section:decomposition}

Subset $S$ of $V(H)$ is a {\em star-cutset}~\cite{Chvatal85} of graph
$H$ if $S\subseteq N_H[s]$ holds for some node $s$ of $S$ and the
number of connected components of $H\setminus S$ is larger than that
of $H$.

\begin{lemma}
\label{lemma:star-cutset}
For any tracker $T=(H,u_1u_2u_3)$ of an $n$-node $m$-edge beetle-free
connected graph $G$, it takes~$O(mn^3)$ time to complete one of the
following three tasks.  Task~1: Ensuring that $H$ contains even holes.
Task~2: Ensuring that $T$ is not lucky.  Task~3: Obtaining an induced
subgraph $H'$ of $H$ having no star-cutsets such that if $T$ is lucky,
then $H'$ contains even holes.
\end{lemma}

\begin{lemma}
\label{lemma:no-star-cutsets}
It takes $O(mn^4)$ time to determine if an $n$-node $m$-edge graph
having no star-cutsets contains even holes.
\end{lemma}

\begin{proof}[Proof of Lemma~\ref{lemma:decomposition}]
We apply Lemma~\ref{lemma:star-cutset} on the input tracker
$T=(H,u_1u_2u_3)$ of $G$ in $O(m n^3)$ time.  If Task~1 or~2 is
completed, then the lemma is proved.  If Task~3 is completed, then 
since $H'$ has no star-cutsets, Lemma~\ref{lemma:no-star-cutsets} 
implies that it takes $O(mn^4)$ time to determine whether $H'$ 
contains even holes.  Since $H'$ is an induced subgraph of $H$, if 
$H'$ contains even holes, then so does $H$; otherwise, $T$ is not 
lucky.
\end{proof}

Subsection~\ref{section:star_cut} proves
Lemma~\ref{lemma:star-cutset}.  Subsection~\ref{section:2-join} proves
Lemma~\ref{lemma:no-star-cutsets}.


\subsection{Proving Lemma~\ref{lemma:star-cutset}}
\label{section:star_cut}

A star-cutset $S$ of graph $H$ is {\em full} if $S=N_H[s]$ holds for
some node $s$ of $S$.  Full star-cutsets in an $n$-node $m$-edge graph
can be detected in $O(mn)$ time.  Node $x$ {\em dominates} node~$y$ in
graph $H$ if $x \neq y$ and~$N_H[y] \subseteq N_H[x]$.  Node $y$ is
{\em dominated} in $H$ if some node of $H$ dominates $y$ in $H$.  We
need the following three lemmas to prove
Lemma~\ref{lemma:star-cutset}.

\begin{lemma}[{Chv\'{a}tal~\cite[Theorem~1]{Chvatal85}}]
\label{lemma:fullstar}
A graph having no dominated nodes and full star-cutsets has no
star-cutsets.
\end{lemma}

\begin{lemma}
\label{lemma:dominate-find}
If $T=(H,u_1u_2u_3)$ is a tracker of an $n$-node $m$-edge beetle-free
connected graph $G$, then it takes~$O(mn^2)$ time to obtain a tracker
$T'=(H',u'_1u'_2u'_3)$ of $G$, where $H'$ is an induced subgraph of
$H$ having no dominated nodes, such that if $T$ is lucky, then so is
$T'$.
\end{lemma}

\begin{proof}
\begin{figure}[t]
\centerline{\scalebox{0.9}{\input{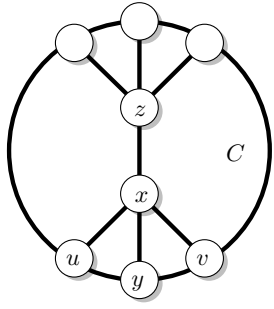}}}
\caption{An illustration for the proof of
  Lemma~\ref{lemma:dominate-find}.}
\label{figure:fig5}
\end{figure}
We first prove the following claim for any beetle-free graph $H$: {\em
  ``If a node $x$ of $H$ dominates a node~$y$ of a clean shortest even
  hole $C$ of $H$, then $C'=H[C \cup \{x\} \setminus \{y\}]$ is a
  clean shortest even hole of $H$.''}  Let $u$ and $v$ be the
neighbors of $y$ on $C$.  Since $C$ is a hole and $y\in C$, we
know~$x\notin C$, implying~$x\in N_H(C)$.  Since $x$ dominates $y$ and
$|N_C[y]| = 3$, there is a connected component of~$C[N_C(x)]$ having
at least $3$ nodes.  By Lemma~\ref{lemma:hole-node}, we have $x\in
N_H^{3}(C)$, implying $N_C(x)=\{u,y,v\}$.  Thus, $C'$ is a shortest
even hole of $H$.  Assume $z\in M_H(C') \cup N_H^{2,2}(C')$ for
contradiction.  By~$y\in N_H^3(C')$, $z\neq y$.  By $C \setminus \{y\}
= C' \setminus \{x\}$, exactly one of $x$ and $y$ is adjacent to $z$
in $H$ or else $z\in M_H(C)\cup N_H^{2,2}(C)$, contradicting the fact
that $C$ is clean.
{\em Case~1: $z \in N_H^{2,2}(C')$}.
If~$z\in N_H(y)\setminus N_H(x)$, then we have~$z\in M_H(C)$,
contradicting the assumption that $C$ is a clean hole of $H$.
If~$z\in N_H(x)\setminus N_H(y)$, then~$z\in N_H^{1,2}(C)$,
contradicting Lemma~\ref{lemma:hole-node}.
{\em Case~2: $z \in M_H(C')$}. 
By $|N_{C'}(z)|\geq 3$ and Lemma~\ref{lemma:hole-parity}, $|N_{C'}(z)|
\ge 4$.  By $M_H(C) = N_H^{2,2}(C) = \varnothing$ and
Lemma~\ref{lemma:hole-node}, $|N_{C}(z)| \le 3$.  By $C\setminus
\{x\}=C\setminus\{y\}$, we have $z\in N_H(x)\setminus N_H(y)$,
$|N_C(z)|=3$, and $|N_{C'}(z)|=4$.  By Lemma~\ref{lemma:hole-node},
$z\in N_H^{3}(C)$. See Figure~\ref{figure:fig5} for an
illustration. Thus, $C[N_C(z)]$ is a $3$-path, implying that $H[C'\cup
  \{z\}]$ is a beetle $B$ of $H$ in which $B[N_B[z]]$ is a diamond, a
contradiction.  The claim is proved.

The algorithm first iteratively updates $(H,u_1u_2u_3)$ by the
following steps until $H$ has no dominated nodes, and then
outputs the resulting $(H,u_1u_2u_3)$ as $(H',u'_1u'_2u'_3)$.
\begin{enumerate}[\em {Step}~1:]
\addtolength{\itemsep}{-0.5\baselineskip}
\item 
\label{dominated1}
Let $x$ and $y$ be two nodes of $H$ such that $x$ dominates $y$.

\item
\label{dominated2}
If there is an $i\in \{1,2,3\}$ with $y=u_i$, then let $u_i=x$.

\item 
\label{dominated3}
Let $H=H \setminus \{y\}$.
\end{enumerate}
It takes $O(mn)$ time to detect nodes $x$ and $y$ such that $x$
dominates $y$. Each iteration of the loop decreases $|V(H)|$ by one
via Step~\ref{dominated3}.  Therefore, the overall running time is
$O(m n^2)$.  Graph $H'$ is an induced subgraph of the initial
$H$. $H'$ has no dominated nodes. It suffices to ensure that if the
tracker $T=(H,u_1u_2u_3)$ of $G$ at the beginning of an iteration is
lucky, then the tracker at the end of the iteration, denoted
$T'=(H',u'_1u'_2u'_3)$, remains lucky.  Let $C$ be a $u_1u_2u_3$-hole
of $H$.  If $y\notin C$, then $C$ is a $u'_1u'_2u'_3$-hole of
$H'=H\setminus \{y\}$.  If $y\in C$, then the claim above ensures that
$C' = H[C \cup \{x\} \setminus \{y\}]$ is a clean shortest even hole
of $H$.  Since $x$ dominates $y$, $u'_1u'_2u'_3$ is a path of hole
$C'$. Thus, $C'$ is a $u'_1u'_2u'_3$-hole of~$H'$.  Either way,
$(H',u'_1u'_2u'_3)$ is lucky.
\end{proof}

\begin{lemma}
\label{lemma:block}
If $(H,u_1u_2u_3)$ is a lucky tracker of graph $G$ and $S$ is a full
star-cutset of $H$, then one of the following two conditions holds:
\begin{enumerate}[{Condition~B}1:]
\addtolength{\itemsep}{-0.5\baselineskip}
\item 
\label{block1}
For each $u_1u_2u_3$-hole $C$ of $H$, there exists a connected
component $B$ of $H\setminus S$ satisfying $C\subseteq H[B\cup S]$.

\item 
\label{block2}
There are two non-adjacent nodes $s_1$ and $s_2$ of $S$ and two
connected components $B_1$ and $B_2$ of $H\setminus S$ with
$\{s_1,s_2\}\subseteq N_H(B_1)$ and $\{s_1,s_2\}\subseteq N_H(B_2)$.
\end{enumerate}
\end{lemma}

\begin{proof}
Let $s$ be a node of $S$ with $N_H[s]=S$.  Let $C$ be a
$u_1u_2u_3$-hole of $H$.  Assume that Condition~B\ref{block1} does not
hold with respect to $C$. There exist two distinct connected
components $B_1$ and $B_2$ of $H\setminus S$ such that $V(C)\cap
V(B_1)\ne\varnothing$ and $V(C)\cap V(B_2)\ne\varnothing$.  Thus,
$C[S]$ has at least two connected components.  Let $s_1$ and $s_2$ be
two nodes in distinct connected components of $C[S]$.  By
$\{s_1,s_2\}\subseteq N_H[s]$, we have $s\notin C$ or else $s$, $s_1$,
and $s_2$ are in the same connected component of $C[S]$.  By
Lemma~\ref{lemma:hole-node}, we have $s\in N_H^{1,1}(C)$, implying
$\{s_1,s_2\}=V(C)\cap S$. It follows that both $s_1$ and $s_2$ are
adjacent to both $B_1$ and $B_2$. Let paths $P_1$ and $P_2$ be the two
connected components of $C\setminus \{s_1,s_2\}$. One of $P_1$ and
$P_2$ has to be in $B_1$ and the other of $P_1$ and $P_2$ has to be in
$B_2$. Therefore, Condition~B\ref{block2} holds.
\end{proof}

\begin{proof}[Proof of Lemma~\ref{lemma:star-cutset}]
Let $T_0$ be the initial given tracker $(H,u_1u_2u_3)$ of $G$.  The
algorithm iteratively updates $(H,u_1u_2u_3)$ by the following three
steps until Task~1,~2, or~3 is completed.
\begin{enumerate}[\em {Step}~1:]
\addtolength{\itemsep}{-0.5\baselineskip}
\item 
\label{starcut1}
Apply Lemma~\ref{lemma:dominate-find} in $O(mn^2)$ time on tracker
$T=(H,u_1u_2u_3)$ to obtain a tracker $T'=(H',u'_1u'_2u'_3)$ of $G$,
where $H'$ is an induced subgraph of $H$ having no dominated nodes,
such that if $T$ is lucky, then so is $T'$.  Determine in $O(mn)$ time
whether $H'$ has full star-cutsets.  If $H'$ has full star-cutsets,
then let $(H,u_1u_2u_3)=(H',u'_1u'_2u'_3)$ and proceed to
Step~\ref{starcut2}; Otherwise, complete Task~3 by outputting $H'$.

\item 
\label{starcut2}
Let $S$ be a full star-cutset of $H$.  If Condition~B\ref{block2} of
Lemma~\ref{lemma:block} holds, then complete Task~1 by outputting that
$G$ contains even holes.  Otherwise, proceed to Step~\ref{starcut3}.

\item 
\label{starcut3}
If either one of the following statements hold for $U=\{u_1,u_2,u_3\}$:
\begin{itemize}
\addtolength{\itemsep}{-0.25\baselineskip}
\item 
$U\subseteq S$ and a connected component $B$ of $H\setminus S$ is
  adjacent to both $u_1$ and $u_3$;

\item 
$U\not\subseteq S$ and $U\subseteq B\cup S$ holds for a connected
  component $B$ of $H\setminus S$,
\end{itemize}

\noindent
then let $H=H[B\cup S]$ and proceed to the next iteration of the loop.
Otherwise, complete Task~2 by outputting that $T_0$ is not lucky.
\end{enumerate}
Step~\ref{starcut1} does not increase $|V(H)|$.  If
Step~\ref{starcut3} updates $H$, then $|V(H)|$ is decreased by at
least one, since $H\setminus S$ has more than one connected component.
The algorithm halts in $O(n)$ iterations.  Step~\ref{starcut1} takes
$O(mn^2)$ time. Step~\ref{starcut2} takes $O(mn^2)$ time: For any two
non-adjacent nodes $s_1$ and $s_2$ in $S$, it takes $O(m)$ time to
determine whether $s_1$ and $s_2$ have two or more common neighboring
connected components of $H\setminus S$.  Step~\ref{starcut3} takes
$O(m)$ time.  The overall running time is $O(mn^3)$.

\begin{figure}[t]
\centerline{\scalebox{0.9}{\input{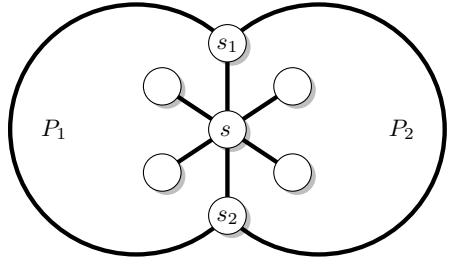}}}
\caption{An illustration for the proof of
  Lemma~\ref{lemma:star-cutset}.}
\label{figure:fig6}
\end{figure}

We first show the following claim for each iteration of the algorithm:
{\em ``If the $(H,u_1u_2u_3)$ at the beginning of an iteration is a
  lucky tracker of $G$, then (1) the intermediate $(H,u_1u_2u_3)$
  throughout the iteration remains a lucky tracker of $G$, and (2)
  Step~\ref{starcut3}, if reached, proceeds to the next iteration.''}
It suffices to consider the situation that Step~\ref{starcut3} is
reached and focus on the update operation that replaces $H$ with
$H[B\cup S]$ via Step~\ref{starcut3}.  By definition of
Step~\ref{starcut2}, Condition~B\ref{block2} does not hold.  By
Lemma~\ref{lemma:block}, Condition~B\ref{block1} holds. That is, some
$u_1u_2u_3$-hole $C$ of $H$ is in a connected component $B^*$ of
$H\setminus S$. We prove the claim by showing that $B^*$ has to be the
connected component $B$ of $H\setminus S$ in Step~\ref{starcut3}.
Since $B^*=B$ holds trivially for the case $\{u_1,u_2,u_3\}
\not\subseteq S$, we assume $\{u_1,u_2,u_3\}\subseteq S$.  If $s\in
C$, then exactly two nodes of $C$ are adjacent to $s$ in $H$; and
otherwise, Lemma~\ref{lemma:hole-node} implies that $s$ has at most
three neighbors of~$H$ in $C$.  Either way, we have $|V(C)\cap S|\leq
3$.  Since $u_1u_2u_3$ is a path of even hole $C$, nodes $u_1$ and
$u_3$ are not adjacent in $H$.  Since Condition~B\ref{block2} does not
hold, at most one connected component of~$H\setminus S$ can be
adjacent to both~$u_1$ and $u_3$ in $H$. By $V(C)\subseteq B^*\cup S$
and $|V(C)\cap S|\leq 3$, we have $(N_C(u_1)\cup N_C(u_3))\setminus
\{u_2\}\subseteq B^*$, implying $B^*=B$.  The claim is proved.

For the correctness of the algorithm, we consider the three possible
steps via which the algorithm halts.  Step~\ref{starcut1}:~Since $H'$
has no dominated nodes and full-star-cutsets,
Lemma~\ref{lemma:fullstar} implies that $H'$ has no star-cutsets.  By
the claim above, Task~3 is completed.
Step~\ref{starcut2}:~Condition~B\ref{block2} holds.  Let $P_1$ be a
shortest path between $s_1$ and $s_2$ in $H[B_1\cup \{s_1,s_2\}]$.
Let $P_2$ be a shortest path between~$s_1$ and $s_2$ in $H[B_2\cup
  \{s_1,s_2\}]$.  Since $s_1$ and $s_2$ are not adjacent, at least one
of the three cycles of graph $P_1\cup P_2\cup \{ss_1,ss_2\}$ is an
even hole of $H$. Since $H$ is an induced subgraph of $G$, $G$
contains even holes.  See Figure~\ref{figure:fig6} for an
illustration.  Task~1 is completed.
Step~\ref{starcut3}:~By the claim above, if $T_0$ is lucky, then
Step~\ref{starcut3} always proceeds to the next iteration of the
loop. Thus, Task~2 is completed.
\end{proof}


\subsection{Proving Lemma~\ref{lemma:no-star-cutsets}}
\label{section:2-join}

\subsubsection{Extended clique trees}
Graph $H$ is an {\em extended clique tree}~\cite{daSilvaV13} if there
is a set $S$ of two or less nodes of $H$ such that each biconnected
component of $H\setminus S$ is a clique.  da~Silva
and~Vu\v{s}kovi\'{c}~\cite[\S2.3]{daSilvaV13} described
an~$O(n^5)$-time algorithm to determine whether an $n$-node extended
clique tree contains even holes, which can actually be implemented to
run in $O(n^4)$ time.

\begin{lemma}
\label{lemma:faster-extended-clique-tree}
\label{lemma:extended-clique-tree-even-hole}
It takes $O(n^4)$ time to determine whether an $n$-node
extended clique tree contains even holes.
\end{lemma}

\begin{proof}
Let $H_0$ be the $n$-node extended clique tree.  Let $x$ and $y$ be
two nodes of $H_0$ such that each biconnected component of
$H=H_0\setminus\{x,y\}$ is a clique.  For nodes $u$ and $v$ of $H$,
let $P(u,v)$ be the shortest path of $H$ between $u$ and $v$ and let
$p(u,v)$ be the number of edges in $P(u,v)$.  We spend~$O(n^4)$ time
to store the following information in a table $M_1$ for every two
nodes $u$ and $v$ of~$H$:
(i)~$p(u,v)$ and (ii) whether or not $P(u,v)\setminus\{u,v\}$ is
adjacent to $x$ (respectively, $y$).  With $M_1$, it takes $O(n^2)$
time to determine whether $H_0$ contains an even hole that passes $y$
but not $x$: $H_0\setminus\{x\}$ contains even holes if and only if
there are two non-adjacent neighbors $u$ and $v$ of $y$ in $H$ such
that~$p(u,v)$ is even and $P(u,v)\setminus\{u,v\}$ is not adjacent to
$y$.
Similarly, with $M_1$, it takes $O(n^2)$ time to determine whether
$H_0$ contains an even hole that passes $x$ but not $y$.

To determine whether $H_0$ contains an even hole that passes both $x$
and $y$, we store in a table $M_2$ for every four nodes
$u_1,v_1,u_2,v_2$ whether or not $P(u_1,v_1)$ and $P(u_2,v_2)$ are
both disjoint and non-adjacent.  It takes $O(n^2)$ time to compute the
connected components of $H\setminus
N_H[P(u_1,v_1)]$. Paths~$P(u_1,v_1)$ and $P(u_2,v_2)$ are both
disjoint and non-adjacent if and only if $u_2$ and $v_2$ are in the
same connected component of $H\setminus N_H[P(u_1,v_1)]$.  Therefore,
$M_2$ can also be computed in $O(n^4)$ time.  With tables $M_1$ and
$M_2$, it takes $O(n^4)$ time to determine whether $H_0$ contains an
even hole that passes both $x$ and $y$:
{\em Case~1: $x$ and $y$ are adjacent in $H_0$.}  $H_0$ contains an
even hole that passes both~$x$ and $y$ if and only if there are nodes
$u$ and $v$ such that (1) $H_0[\{u,x,y,v\}]$ is path $uxyv$, (2)
$p(u,v)$ is odd, and (3)~$P(u,v)\setminus \{u,v\}$ is not adjacent to
$\{x,y\}$.
{\em Case~2: $x$ and $y$ are not adjacent in $H_0$.}  $H_0$ contains
an even hole that passes both $x$ and $y$ if and only if there are
nodes $u_x,v_x,u_y,v_y$ of $H$ such that (1)~$H_0[\{u_x,x,v_x\}]$ is
path $u_xxv_x$ and $H_0[\{u_y,y,v_y\}]$ is path $u_yyv_y$, (2)
$p(u_x,u_y)+p(v_x,v_y)$ is even, and (3) $P(u_x,u_y)$ and $P(v_x,v_y)$
are both disjoint and non-adjacent.
\end{proof}

\subsubsection{2-joins and non-path 2-joins}

\begin{figure}[t]
\centerline{\scalebox{0.9}{\input{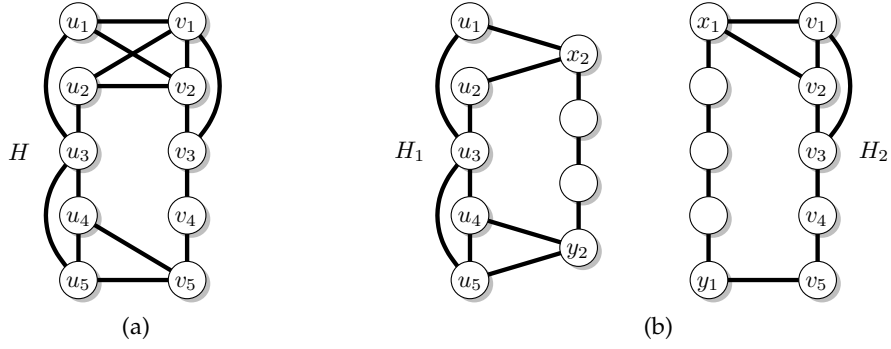}}}
\caption{(a) A connected non-path $2$-join $V_1| V_2$ of graph $H$
  with split $(X_1,Y_1,X_2,Y_2)$, where $X_1=\{u_1,u_2\}$,
  $X_2=\{v_1,v_2\}$, $Y_1=\{u_4,u_5\}$, $Y_2=\{v_5\}$, $V_1=X_1\cup
  Y_1\cup \{u_3\}$, and $V_2=X_2\cup Y_2\cup \{v_3,v_4\}$. (b) The
  parity-preserving blocks $H_1$ and $H_2$ of decomposition of $H$ for
  the connected $2$-join $V_1| V_2$ with respect to the split
  $(X_1,Y_1,X_2,Y_2)$.}
\label{figure:fig7}
\end{figure}

We say that $V_1|V_2$ is a {\em
  $2$-join}~\cite{CornuejolsC85,TrotignonV12} of a graph $H$ with {\em
  split} $(X_1,Y_1,X_2,Y_2)$ if (1) $V_1$ and $V_2$ form a disjoint
partition of $V(H)$ with $|V_1|\geq 3$ and $|V_2|\geq 3$, (2) $X_1$
and $Y_1$ (respectively, $X_2$ and $Y_2$) are disjoint non-empty
subsets of $V_1$ (respectively, $V_2$), and (3) each node of $X_1$ is
adjacent to each node of $X_2$, each node of $Y_1$ is adjacent to each
node of $Y_2$, and there are no other edges between~$V_1$ and
$V_2$. See Figure~\ref{figure:fig7}(a) for an example.

\begin{lemma}[Trotignon et al.~{\cite[Lemma~3.2]{TrotignonV12}}]
\label{lemma:is-connected}
Let $H$ be a graph having no star-cutsets.  If $V_1|V_2$ is a $2$-join
of $H$ with split $(X_1,Y_1,X_2,Y_2)$, then all of the following four statements
hold for each $i\in \{1,2\}$.
\begin{enumerate}
\addtolength{\itemsep}{-0.5\baselineskip}
\item Each connected component of $H[V_i]$ has at least one node
  in $X_i$ and at least one node in $Y_i$.

\item Each node of $V_i$ has a neighbor in $V_i$.

\item Each node of $X_i$ has a non-neighbor in $Y_i$.  Each node of
  $Y_i$ has a non-neighbor in $X_i$.

\item $|V_i|\geq 4$.
\end{enumerate}
\end{lemma}

A $2$-join $V_1|V_2$ of $H$ with split $(X_1,Y_1,X_2,Y_2)$ is a {\em
  non-path $2$-join}~\cite{Trotignon08} of $H$ if $H[V_1]$ is not a
path between a node of $X_1$ and a node of $Y_1$ and $H[V_2]$ is not a
path between a node of $X_2$ and a node of~$Y_2$.
For instance, the $2$-join in Figure~\ref{figure:fig7}(a) is a
non-path $2$-join.
(Non-path $2$-joins are called~$2$-joins by da~Silva and
Vu\v{s}kovi\'{c}~\cite[\S1.3]{daSilvaV13}.)

\begin{lemma}[Charbit et al.~{\cite[Theorem~4.1]{CharbitHTV12}}]
\label{lemma:non-path-2-join}
Given an $n$-node connected graph $H$, it takes $O(n^4)$ time to
either output a non-path $2$-join of $H$ together with a split or
ensure that $H$ has no non-path $2$-joins.
\end{lemma}

\begin{theorem}[da~Silva and~Vu\v{s}kovi\'{c}~{\cite[Corollary~1.3]{daSilvaV13}}]
\label{lemma:structure-new}
A connected even-hole-free graph that has no star-cutsets and
non-path-$2$-joins is an extended clique tree.
\end{theorem}

Combining Lemmas~\ref{lemma:faster-extended-clique-tree}
and~\ref{lemma:non-path-2-join} and Theorem~\ref{lemma:structure-new},
we have the following lemma.
\begin{lemma}
\label{lemma:non-path-2-join-free}
Given an $n$-node graph $H$ having no star-cutsets, it takes $O(n^4)$
time to either (a) determine whether $H$ contains even holes or (b)
obtain a non-path $2$-join of $H$ with a split.
\end{lemma}

\begin{proof}
It takes $O(n^4)$ time to determine whether the graph $H$ is an
extended clique tree: For any set~$S$ of two or less nodes of $H$, it
takes $O(n^2)$ time to obtain the biconnected components of
subgraph~$H\setminus S$~\cite{HopcroftT73} and determine whether all
of them are cliques.  If $H$ is an extended clique tree, then
Lemma~\ref{lemma:faster-extended-clique-tree} implies that it takes
$O(n^4)$ time to determine whether $H$ contains even holes.  If $H$ is
not an extended clique tree, then Lemma~\ref{lemma:non-path-2-join}
implies that it takes $O(n^4)$ time to either obtain a
non-path~$2$-join of $H$ with a split or ensure that $H$ has no
non-path $2$-joins.  If $H$ has no non-path $2$-joins, then
Theorem~\ref{lemma:structure-new} implies that $H$ contains even
holes.
\end{proof}

\subsubsection{Parity-preserving blocks of decomposition for connected 2-joins}

A $2$-join $V_1| V_2$ with split $(X_1,Y_1,X_2,Y_2)$ is {\em
  connected}~\cite{TrotignonV12} if, for each $i \in \{1,2\}$, there
is an induced path $P_i$ of $H[V_i]$ between a node $x_i$ of $X_i$ and
a node $y_i$ of $Y_i$ such that $V(P_i)\setminus\{x_i,y_i\}\subseteq
V_i\setminus (X_i\cup Y_i)$.  For instance, the $2$-join $V_1| V_2$ in
Figure~\ref{figure:fig7}(a) is connected.  By
Lemma~\ref{lemma:is-connected}(1), any $2$-join of a graph having no
star-cutsets is connected with respect to any split.

Let $V_1| V_2$ be a connected $2$-join of graph $H$ with split
$(X_1,Y_1,X_2,Y_2)$.  For each $i \in \{1,2\}$, let $P_i$ be a
shortest induced path $P_i$ of $H[V_i]$ between a node $x_i$ of $X_i$
and a node $y_i$ of $Y_i$ with $V(P_i)\setminus\{x_i,y_i\}\subseteq
V_i\setminus (X_i\cup Y_i)$.  If $|V(P_i)|$ is even (respectively,
odd), then let $p_i=4$ (respectively, $p_i=5$).  The {\em
  parity-preserving blocks of decomposition}~\cite{TrotignonV12} of
$H$ for $2$-join $V_1| V_2$ with respect to split $(X_1,Y_1,X_2,Y_2)$
are the following graphs $H_1$ and $H_2$.
\begin{itemize}
\addtolength{\itemsep}{-0.5\baselineskip}
\item $H_1$ consists of (a) $H[V_1]$, (b) a $p_2$-path between nodes
  $x_2$ and $y_2$, (c) edges $x_2x$ for all nodes $x$ of $X_1$, and
  (d) edges $y_2y$ for all nodes $y$ of $Y_1$.
\item $H_2$ consists of (a) $H[V_2]$, (b) a $p_1$-path between nodes
  $x_1$ and $y_1$, (c) edges $x_1x$ for all nodes $x$ of $X_2$, and
  (d) edges $y_1y$ for all nodes $y$ of $Y_2$.
\end{itemize}
See Figure~\ref{figure:fig7}(b) for an example of $H_1$ and $H_2$.

\begin{lemma}[Trotignon and Vu\v{s}kovi\'{c}~{\cite[Lemma~3.8]{TrotignonV12}}]
\label{lemma:if-and-only-if}
If $V_1| V_2$ is a connected $2$-join of a graph $H$ having no
star-cutsets with split $(X_1,Y_1,X_2,Y_2)$, then the
parity-preserving blocks $H_1$ and $H_2$ of decomposition of $H$
for~$V_1| V_2$ with respect to $(X_1,Y_1,X_2,Y_2)$ are graphs having
no star-cutsets such that $H$ is even-hole-free if and only if both
$H_1$ and $H_2$ are even-hole-free.
\end{lemma}

\begin{lemma}
\label{lemma:fewer-edges}
Let $H$ be an $n$-node $m$-edge graph having no star-cutsets.  Both of
the parity-preserving blocks~$H_1$ and $H_2$ of decomposition for an
arbitrary non-path $2$-join of $H$ with respect to any split have at
most $n$ nodes and $m-1$ edges.
\end{lemma}

\begin{proof}
We prove the lemma for $H_1$. The proof for $H_2$ is similar.  Let
$V_1|V_2$ be the non-path $2$-join.  Let $(X_1,Y_1,X_2,Y_2)$ be the
split.  Let $P_2$ be a shortest path of $H[V_2]$ between a node of
$X_2$ and a node of $Y_2$.
For the case that $|V(P_2)|$ is even, we have $p_2=4$. By
Lemma~\ref{lemma:is-connected}(4), $|V_2|\geq 4$, implying
$|V(H_1)|=n-|V_2|+p_2\leq n$. By the following case analysis, $H_1$
has at most $m-1$ edges.
\begin{itemize}
\addtolength{\itemsep}{-0.5\baselineskip}
\item $|V(P_2)| \geq 6$: By $P_2\subseteq H[V_2]$, $H[V_2]$ has at
  least five edges. Thus, $H_1$ has at most $m-2$ edges.
\item $|V(P_2)|=4$: Since $V_1|V_2$ is a non-path $2$-join of $H$,
  $P_2\subsetneq H[V_2]$. If $V(P_2)=V_2$, then $H[V_2]$ has at least
  four edges.  If $V(P_2)\subsetneq V_2$, then
  Lemma~\ref{lemma:is-connected}(2) implies that $H[V_2]$ has at least
  four edges.  Either way, $H_1$ has at most $m-1$ edges.
\item $|V(P_2)|=2$: Lemma~\ref{lemma:is-connected}(3) ensures
  $|X_2|\geq 2$ and $|Y_2|\geq 2$.  Lemma~\ref{lemma:is-connected}(1)
  implies that $H[V_2]$ has at least two edges.  By $|X_2|\geq 2$ and
  $|Y_2|\geq 2$, the number of edges between $V_1$ and $V_2$ in $H$ is
  at least two more than the number of edges between $V_1$ and
  $V(H_1)\setminus V_1$ in $H_1$.  Therefore, $H_1$ has at most $m-1$
  edges.
\end{itemize}
As for the case that $|V(P_2)|$ is odd, we have $p_2=5$.  The
following case analysis shows that $H_1$ has at most $n$ nodes and at
most $m-1$ edges.
\begin{itemize}
\addtolength{\itemsep}{-0.5\baselineskip}
\item $|V(P_2)| \geq 5$: By $|V_2|\geq 5$, we have $|V(H_1)|\leq n$.
  $P_2$ has at least four edges. Since $V_1|V_2$ is a non-path
  $2$-join of $H$, $P_2\subsetneq H[V_2]$. If $V(P_2)=V_2$, then
  $H[V_2]$ has at least five edges.  If $V(P_2)\subsetneq V_2$, then
  Lemma~\ref{lemma:is-connected}(2) implies that $H[V_2]$ has at least
  five edges.  Either way, $H_1$ has at most $m-1$ edges.

\item $|V(P_2)|=3$: By Lemma~\ref{lemma:is-connected}(4), the proper
  subset $Z=V_2\setminus V(P_2)$ of $V_2$ is non-empty. We know $Z\cap
  (X_2\cup Y_2)\ne\varnothing$ or else $V(P_2)$ would be a star-cutset
  of $H$. Assume $Z\cap X_2\ne\varnothing$ without loss of generality.
  Let $B$ be an arbitrary connected component of $H[Z]$ with $B\cap
  X_2\ne\varnothing$. We know that $B$ is adjacent to $Y_2$ in $H$ or
  else $N_H[x]\setminus Z$ would be a star-cutset of $H$, where~$x$ is
  the endpoint of $P_2$ in $X_2$.  Since $P_2$ is a shortest path
  between a node of $X_2$ and a node of $Y_2$, at least one node of
  $B$ is not in $X_2\cup Y_2$.  Therefore, $|V_2|\geq 5$, implying
  $|V(H_1)|\leq n$.  Moreover,~$H[V_2]$ has at least four edges.  By
  $|X_2|\geq 2$, the number of edges between $V_1$ and $V_2$ in $H$ is
  at least one more than the number of edges between $V_1$ and
  $V(H_1)\setminus V_1$ in $H_1$.  Thus, $H_1$ has at most $m-1$
  edges.
\end{itemize}
\end{proof}

\subsubsection{Proving Lemma~\ref{lemma:no-star-cutsets}}
We now prove Lemma~\ref{lemma:no-star-cutsets} by
Lemmas~\ref{lemma:non-path-2-join-free},~\ref{lemma:if-and-only-if}
and~\ref{lemma:fewer-edges}.
\begin{proof}[Proof of Lemma~\ref{lemma:no-star-cutsets}]
Assume without loss of generality that the given $n$-node $m$-edge
graph $H_0$ having no star-cutsets is connected.  Let set $\HH$
initially consist of a single graph $H_0$.  We then repeat the
following loop until $\HH=\varnothing$ or we output that $H_0$
contains even holes.  Let $H$ be a graph in $\HH$.
{\em Case~1: $H$ has at most $11$ edges.}  It takes $O(1)$ time
to determine whether $H$ contains even holes. If $H$ contains even
holes, then we output that $H_0$ contains even holes.  Otherwise, we
delete $H$ from $\HH$.
{\em Case~2: $H$ has at least $12$ edges.}  We first delete $H$
from $\HH$ and then apply Lemma~\ref{lemma:non-path-2-join-free} on
$H$.  If $H$ contains even holes, then we output that $H_0$ contains
even holes.  If we obtain a non-path $2$-join $V_1|V_2$ of $H$ with
split $(X_1,Y_1,X_2,Y_2)$, then we add to $\HH$ the parity-preserving
blocks $H_1$ and $H_2$ of decomposition for $V_1|V_2$ with respect to
$(X_1,Y_1,X_2,Y_2)$.  If the above loop stops with $\HH=\varnothing$,
then we output that $H_0$ is even-hole-free.

The correctness of our algorithm follows immediately from
Lemma~\ref{lemma:if-and-only-if}.  By Lemma~\ref{lemma:fewer-edges},
each graph ever in $\HH$ throughout our algorithm has at most $n$
nodes.  By Lemma~\ref{lemma:non-path-2-join-free}, each iteration of
the loop takes $O(n^4)$ time.  It remains to show that the loop halts
in $O(m)$ iterations.  Observe that each iteration increases the
overall number of edges of the graphs in $\HH$ by no more than $10$.
Let~$f(m)$ be the maximum number of iterations of the above loop in
which Lemma~\ref{lemma:non-path-2-join-free} is
applied. Lemma~\ref{lemma:fewer-edges} implies
\begin{displaymath}
f(m)\leq 
\left\{
\begin{array}{ll}
0&\mbox{if $m\leq 11$}\\
\max\{1+f(m_1)+f(m_2): m_1,m_2\leq m-1, m_1+m_2\leq m+10\}&\mbox{if $m\geq 12$}.
\end{array}
\right.
\end{displaymath}
By induction on $m$, we show $f(m)\leq\max(m-11,0)$, which clearly
holds for $m=1,2,\ldots,11$.  As for the induction step, if~$m\geq 12$, then the inductive
hypothesis implies that $f(m)$ is at most
\begin{eqnarray*}
&&\max\{1+\max(m_1-11,0)+\max(m_2-11,0): m_1,m_2\leq m-1, m_1+m_2\leq m+10\}\\
&\leq&\max\{\max(m_1+m_2-21,m_1-10,m_2-10,1): m_1,m_2\leq m-1, m_1+m_2\leq m+10\}\\
&\leq&\max(m-11,m-11,m-11,1)\\
&=&\max(m-11,0).
\end{eqnarray*}
By $f(m)=O(m)$, the number of iterations of the above loop is $O(m)$. 
\end{proof}


\section{Concluding remarks}
\label{section:conclude}

For any family $\mathbbmsl{G}$ of graphs, one can augment a
recognition algorithm for $\mathbbmsl{G}$-free graphs into a
$\mathbbmsl{G}$-detection algorithm for an $n$-node graph $G$ with a
factor-$O(n)$ increase in the time complexity by a node-deletion
method: (1) Let $H=G$. (2) For each node $v$ of $G$, if $H\setminus
\{v\}$ is not $\mathbbmsl{G}$-free, then let $H=H\setminus\{v\}$. (3)
Output the resulting graph~$H$.  See, e.g.,~\cite[\S4]{Vuskovic10} for
the case that $\mathbbmsl{G}$ consists of even holes.  Thus,
Theorem~\ref{theorem:theorem1} immediately yields a detection
algorithm that runs in time $O(m^3 n^6)=O(n^{12})$.  However, our
$O(m^3 n^5)$-time recognition algorithm can be augmented into an
even-hole-detection algorithm without increasing the time complexity.

The combination of the proofs of Theorem~\ref{theorem:theorem1} and
Lemma~\ref{lemma:cleaning} actually gives two algorithms.  The first
algorithm determines if $G$ is both beetle-free and $4$-hole-free.
The second algorithm determines if a beetle-free and $4$-hole-free
graph $G$ is also even-hole-free.  We first describe how to augment
the first algorithm into an $O(m^3n^5)$-time detection algorithm.
Since it takes $O(n^4)$ time to detect a $4$-hole in $G$, it suffices
to show how to detect an even hole in a graph $G$ with beetles in
$O(m^3 n^5)$ time.  As stated in the proof of
Lemma~\ref{lemma:cleaning}, for each of the $O(m^3 n)$ choices of node
$b_4$ and edges $b_1b_5$, $b_2b_6$, and $b_3b_7$, it takes $O(n^4)$
time via Theorem~\ref{lemma:three_tree} to determine if $b_4$,
$b_1b_5$, $b_2b_6$, and $b_3b_7$ are in a beetle $B$ in which
$\{b_1,b_2,b_3,b_4\}$ induces a diamond.  Once we know that a
particular choice of $b_4$, $b_1b_5$, $b_2b_6$, and $b_3$ is in some
beetle $B$, it takes $O(n^5)$ time to actually detect such a
beetle~$B$ by Theorem~\ref{lemma:three_tree} augmented via the
node-deletion method above. Therefore, if $G$ contains beetles, then
it takes $O(m^3 n)\cdot O(n^4)+O(n^5)=O(m^3n^5)$ time to find a beetle
of $G$, in which an even hole of $G$ can be detected in $O(n)$ time.

The second algorithm can also be augmented into an $O(m^3 n^5)$-time
detection algorithm for a beetle-free graph $G$ that contains even
holes.  By Lemma~\ref{lemma:cleaning}, we obtain in $O(m^3 n^5)$ time
a set $\mathbbmsl{T}$ of $O(m^2 n)$ trackers of $G$ that satisfies
Condition~L1.  Since $G$ contains even holes, there must be a tracker
$(H,u_1u_2u_2)$ of $\mathbbmsl{T}$ such that $H$ contains an even hole
of $G$, which according to Lemma~\ref{lemma:decomposition} can be
found in time $O(m^2 n)\cdot O(m n^4)=O(m^3 n^5)$. By the proof of
Lemma~\ref{lemma:decomposition}, $H$ is ensured to contain even holes
in two ways.  (1) If it is ensured through completing Task~1 of
Lemma~\ref{lemma:star-cutset}, then the proof of
Lemma~\ref{lemma:star-cutset} actually gives a constructive proof for
the existence of an even hole of $H$, which is also an even hole of
$G$.  (2) If it is ensured through completing Task~3 of
Lemma~\ref{lemma:star-cutset}, then we have an induced subgraph $H'$
of $H$ having no star-cutsets which is ensured to contain even holes
via Lemma~\ref{lemma:no-star-cutsets}.
We then apply the above node-deletion method on $H'$ using
Lemma~\ref{lemma:no-star-cutsets} to detect in $O(m n^5)$ time an even
hole of $H'$, which is also an even hole of $H$ and $G$. Therefore, if
$G$ is a $4$-hole-free and beetle-free graph that contains even holes,
then it takes time $O(m^3n^5)+O(mn^5)=O(m^3n^5)$ to output an even
hole of~$G$.  Combining the two detection algorithms above, we have an
$O(m^3n^5)$-time algorithm that outputs an even hole in an $n$-node
$m$-edge graph containing even holes.


\section*{Acknowledgments}
We thank Gerard J. Chang for discussion.
We also thank the anonymous reviewers for their helpful comments.

\bibliographystyle{abbrv}
\bibliography{hole}
\end{document}